\documentclass[11pt]{article}
\usepackage{balance}
\usepackage{graphicx}
\usepackage[usenames,dvipsnames]{color}
\usepackage{enumerate}
\usepackage{nicefrac}
\usepackage{algonjk,xspace,amsmath}
\usepackage{hyperref}

\usepackage{amsfonts, amssymb, mathrsfs,amsmath,amsthm}
\newtheorem{definition}{Definition}
\newtheorem{theorem}{Theorem}

\newtheorem{lemma}[theorem]{Lemma}

\newcommand{\er}{Erd\H{o}s-R\'{e}nyi\xspace}
\newcommand{\etal}{\emph{et al.}\xspace}
\newcommand{\eps}{\varepsilon}
\newcommand{\um}{{\sc User-Matching}\xspace}
\newcommand{\comment}[1]{}

\advance \topmargin by -1.0in
\advance \oddsidemargin by -0.8in
\newcommand{\myparskip}{3pt}
\parskip \myparskip

\begin{document}

%\pagenumbering{gobble}

%\title{Tell me who your friends are, and I'll tell you who you are: An efficient reconciliation algorithm for social networks}
\title{An efficient  reconciliation algorithm for social networks}

\author{
  Nitish Korula \thanks{Google Inc., 76 Ninth Ave, 4th Floor, New York,
    NY 10011. \texttt{nitish@google.com}}
  \and
  Silvio Lattanzi \thanks{Google Inc., 76 Ninth Ave, 4th Floor, New York,
    NY 10011. \texttt{silviol@google.com}}
}

\date{}

\maketitle
\begin{abstract}
  People today typically use multiple online social networks
  (Facebook, Twitter, Google+, LinkedIn, etc.). Each online network
  represents a subset of their ``real'' ego-networks.  An interesting
  and challenging problem is to reconcile these online networks, that
  is, to identify all the accounts belonging to the same
  individual. Besides providing a richer understanding of social
  dynamics, the problem has a number of practical applications.
  At first sight, this problem appears algorithmically challenging.
  Fortunately, a small fraction of individuals explicitly link their
  accounts across multiple networks; our work leverages these
  connections to identify a very large fraction of the network.

  Our main contributions are to mathematically formalize the problem
  for the first time, and to design a simple, local, and efficient
  parallel algorithm to solve it. We are able to prove strong
  theoretical guarantees on the algorithm's performance on
  well-established network models (Random Graphs, Preferential
  Attachment). We also experimentally confirm the effectiveness of the
  algorithm on synthetic and real social network data sets.
\end{abstract}

\section{Introduction}

The advent of online social networks has generated a renaissance in
the study of social behaviors and in the understanding of the topology
of social interactions.  For the first time, it has become possible to
analyze networks and social phenomena on a world-wide scale and to
design large-scale experiments on them.  This new evolution in social
science has been the center of much attention, but has also attracted
a lot of critiques; in particular, a longstanding problem in the study
of online social networks is to understand the similarity between them
and ``real'' underlying social networks~\cite{G13}.

This question is particularly challenging because online social
networks are often just a realization of a subset of real social
networks. For example, Facebook ``friends'' are a good representation
of the personal acquaintances of a user, but probably a poor
representation of her working contacts, while LinkedIn is a good
representation of work contacts but not a very good representation of
personal relationships. Therefore, analyzing social behaviors in any
of these networks has the drawback that the results would only be
partial. Furthermore, even if certain behavior can be observed in
several networks, there are still serious problems because there is no
systematic way to combine the behavior of a specific user across
different social networks and because some social relationships will
not appear in any social network. For these reasons, identifying all the
accounts belonging to the same individual across different social
services is a fundamental step in the study of social science.

Interestingly, the problem has also very important practical
implications. First, having a deeper understanding of the
characteristics of a user across different networks helps to construct
a better portrait of her, which can be used to serve personalized
content or advertisements. In addition, having information about
connections of a user across multiple networks would make it easier to
construct tools such as ``friend suggestion'' or ``people you may want
to follow''.

The problem of identifying users across online social networks (also
referred to as the social network reconciliation problem) has been
studied extensively using machine learning techniques; several
heuristics have been proposed to tackle it. However, to the best of
our knowledge, it has not yet been studied formally and no rigorous
results have been proved for it.  One of the contributions of our work
is to give a formal definition of the problem, which is a precursor to
mathematical analysis. Such a definition requires two key components:
A model of the ``true'' underlying social network, and a model for how
each online social network is formed as a subset of this network. We
discuss details of our models in Section~\ref{sec:model}.

Another possible reason for the lack of mathematical analysis is that
natural definitions of the problem are demotivatingly similar to the
graph isomorphism problem.\footnote{In graph theory, an isomorphism
  between two graphs $G$ and $H$ is a bijection, $f(*)$, between the
  vertex sets of $G$ and $H$ such that any two vertices $u$ and $v$ of
  $G$ are adjacent in $G$ if and only if $f(u)$ and $f(v)$ are
  adjacent in $H$. The graph isomorphism problem is: Given two graphs
  $G$ and $G'$ find an isomorphism between them or determine that
  there is no isomorphism. The graph isomorphism problem is considered
  very hard, and no polynomial algorithms are known for it.} In
addition, at first sight the social network reconciliation problem
seems even harder because we are not looking just for isomorphism but
for similar structures, as distinct social networks are not identical.
Fortunately, when reconciling social networks, we have two advantages
over general graph isomorphism: First, real social networks are not
the adversarially designed graphs which are hard instances of graph
isomorphism, and second, a small fraction of social network users
explicitly link their accounts across multiple networks.

The main goal of this paper is to design an algorithm with
\emph{provable guarantees} that is simple, parallelizable and robust
to malicious users. For real applications, this last property is
fundamental, and often underestimated by machine learning
models.\footnote{Approaches based largely on features of a user (such
  as her profile) and her neighbors can easily be tricked by a
  malicious user, who can create a profile locally identical to the
  attacked user.} In fact, the threat of malicious users is so
prominent that large social networks (Twitter, Google+, Facebook) have
introduced the notion of `verification' for celebrities.

Our first contribution is to give a formal model for the graph
reconciliation problem that captures the hardness of the problem and
the notion of an initial set of trusted links identifying users across
different networks. Intuitively, our model postulates the existence of
a true underlying graph, then randomly generates 2 realizations of it
which are perturbations of the initial graph, and a set of trusted
links for some users. Given this model, our next significant
contribution is to design a simple, parallelizable algorithm (based on
similar intuition to the algorithm in~\cite{deanonymizing09}) and to
prove formally that our algorithm solves the graph reconciliation
problem if the underlying graph is generated by well established
network models. It is important to note that our algorithm relies on
graph structure and the initial set of links of users across different
networks in such a way that in order to circumvent it, an attacker
must be able to have a lot of friends in common with the user under
attack. Thus it is more resilient to attack than much of the previous
work on this topic.  Finally, we note that any mathematical model is,
by necessity, a simplification of reality, and hence it is important
to empirically validate the effectiveness of our approach when the
assumptions of our models are not satisfied. In
Section~\ref{sec:experiments}, we measure the performance of our
algorithm on several synthetic and ``real'' data sets.

We also remark that for various applications, it may be possible to
improve on the performance of our algorithm by adding heuristics based
on domain-specific knowledge. For example, we later discuss
identifying common Wikipedia articles across languages; in this
setting, machine translation of article titles can provide an
additional useful signal. However, an important message of this paper
is that a simple, efficient and scalable algorithm that does not take
any domain-specific information into account can achieve excellent
results for mathematically sound reasons.

\section{Related Work}

The problem of identifying Internet users was introduced to identify
users across different chat groups or web sessions in~\cite{nrt04,
  rr00}. Both papers are based on similar intuition, using writing
style (stylography features) and a few semantic features to identify
users.  The social network reconciliation problem was introduced more
recently by Zafarani and Liu in~\cite{zafarani}. The main intuition
behind their paper is that users tend to use similar usernames across
multiple social networks, and even when different, search engines find
the corresponding names. To improve on these first naive approaches,
several machine learning models were developed~\cite{ahhk10,lth11,
  mtmka12,ncm12, rc10}, all of which collect several features of the
users (name, location, image, connections topology), based on which
they try to identify users across networks. These techniques may be
very fragile with respect to malicious users, as it is not hard to
create a fake profile with similar characteristics. Furthermore, they
get lower precision experimentally than our algorithm
achieves. However, we note that these techniques can be combined with
ours, both to validate / increase the number of initial trusted links,
and to further improve the performance of our algorithm.

A different approach was studied in~\cite{mvgd11}, where the authors
infer missing attributes of a user in an online social network from
the attribute information provided by other users in the network.  To
achieve their results, they retrieve communities, identify the main
attribute of a community and then spread this attribute to all the
user in the community. Though it is interesting, this approach suffers
from the same limitations of the learning techniques discussed above.

Recently, Henderson et al.~\cite{hglaetf11} studied which are the most
important features to identify a node in a social network, focusing
only on graph structure information. They analyzed several features of
each ego-network, and also added the notion of recursive features on
nodes at distance larger than 1 from a specific node.  It is
interesting to notice that their recursive features are more resilient
to attack by malicious users, although they can be easily circumvented
by the attacker typically assumed in the social network security
literature~\cite{ykgf08}, who can create arbitrarily many nodes.

The problem of reconciling social networks is closely connected to the
problem of de-anonymizing social networks.  Backstrom \etal introduced
the problem of deanonymizing social networks in~\cite{bdk07}. In their
paper, they present 2 main techniques: An active attack (nodes are
added to the network before the network is anonymized), and a second
passive one. Our setting is similar to that described in their passive
attack. In this setting the authors are able to design a heuristic
with good experimental results; though their technique is very
interesting, it is somewhat elaborate and does not have a provable
guarantee.

In the context of de-anonymizing social networks, the work of
Narayanan and Shmatikov~\cite{deanonymizing09} is closely
related. Their algorithm is similar in spirit to ours; they look at
the number of common neighbors and other statistics, and then they
keep all the links above a specific threshold. There are two main
differences between our work and theirs. First, we formulate the
problem and the algorithm mathematically and we are able to prove
theoretical guarantees for our algorithm. Second, to improve the
precision of their algorithm in~\cite{deanonymizing09} the authors
construct a scoring function that is expansive to compute. In fact the
complexity of their algorithm is $O((E_1+E_2)\Delta_1\Delta_2)$, where
$E_1$ and $E_2$ are the number of edges in the two graphs and
$\Delta_1$ and $\Delta_2$ are the maximum degree in the 2 graphs. Thus
their algorithm would be too slow to run on Twitter and Facebook, for
example; Twitter has more than 200M users, several of whom have degree
more than 20M and Facebook more than 1B users with several users of
degree 5K. Instead, in our work we are able to show that a very simple
technique based on degree bucketing combined with the number of common
neighbors suffices to guarantee strong theoretical guarantees and good
experimental results. In this way we designed an algorithm with
sequential complexity $O((E_1+E_2)min(\Delta_1, \Delta_2)\log
(\max(\Delta_1, \Delta_2)))$ that can be run in $O(\log
(\max(\Delta_1, \Delta_2)))$ MapReduce rounds. In this context, our
paper can be seen as the first really scalable algorithm for network
de-anonymization with theoretical guarantees. Further, we also obtain
considerably higher precision experimentally, though a perfect
comparison across different datasets is not possible. The different
contexts also are important: In de-anonymization, the precision of
72\% they report corresponds to a significant violation of user
privacy. In contrast, we focus on the benefits to users of linking
accounts; in a user-facing application, suggesting an account with a
28\% chance of error is unlikely to be acceptable.

Finally, independently from our work, Yartseva and
Grossglauser~\cite{YG13} recently studied a very similar model
focusing only on networks generated by the \er random
graph model.

\section{Model and Algorithm}
\label{sec:model}

In this section, we first describe our formal model and its
parameters. We then describe our algorithm and discuss the intuition
behind it. 

\subsection{Model}
\label{subsec:model}
Recall that a formal definition of the user identification problem
requires first a model for the ``true'' underlying social network
$G(V, E)$ that captures relationships between people. However, we
cannot directly observe this network; instead, we consider two
imperfect realizations or copies $G_1(V, E_1)$ and $G_2(V, E_2)$ with
$E_1, E_2 \subseteq E$. Second, we need a model for how edges of $E$
are selected for the two copies $E_1$ and $E_2$. This model must
capture the fact that users do not necessarily replicate their entire
personal networks on any social networking service, but only a subset.

Any such mathematical models are necessarily imperfect descriptions of
reality, and as models become more `realistic', they become more
mathematically intractable. In this paper, we consider certain
well-studied models, and provide complete proofs. It is
possible to generalize our mathematical techniques to some
variants of these models; for instance, with small probability, the
two copies could have new ``noise'' edges not present in the original
network $G(V, E)$, or \emph{vertices} could be deleted in the
copies. We do not fully analyze these as the generalizations require
tedious calculations without adding new insights. Our experimental
results of Section~\ref{sec:experiments} show that the algorithm
performs well even in real networks where the formal mathematical
assumptions are not satisfied.

For the underlying social network, our main focus is on the
\emph{preferential attachment} model~\cite{ba99}, which is
historically the most cited model for social networks. Though the
model does not capture some features of real social networks, the key
properties we use for our analysis are those common to online social
networks such as a skewed degree distribution, and the fact that nodes
have distinct neighbors including some long-range / random connections
not shared with those immediately around them\cite{G83,K00}. In the
experimental section we will consider also different models and also real social
networks as our underline real networks.

For the two imperfect copies of the underlying network we assume that
$G_1$ (respectively $G_2$) is created by selecting each edge $e \in E$
of the original graph $G(V,E)$ independently with a fixed probability
$s_1$ (resp. $s_2$) (See Figure~\ref{fig:model}.) In the real world,
edges/relationships are not selected truly independently, but this
serves as a reasonable approximation for observed networks.  In fact,
a similar model has been previously considered by \cite{pg11}, which
also produced experimental evidence from an email network to support
the independent random selection of edges. Another plausible mechanism
for edge creation in social network is the \emph{cascade}
model, in which nodes are more likely to join a new network if more of
their friends have joined it. Experimentally, we show that our
algorithm performs even better in the cascade model than in the
independent edge deletion model.

These two models are theoretically interesting and practically
interesting~\cite{pg11}. Nevertheless, in some cases the analyzed
social networks may differ in their scopes and so the group of friends
that a user has in a social network can greatly differ from the group
of friends that same user has in the other network. To capture this
scenario in the experimental section, we also consider the Affiliation
Network model~\cite{ls09} (in which users participate in a number of
\emph{communities}) as the underlying social network. For each of
$G_1, G_2$, and for each community, we keep or delete all the edges
inside the community with constant probability. This highly correlated
edge deletion process captures the fact that a user's personal friends
might be connected to her on one network, while her work colleagues
are connected on the second network.  We defer the detailed
description of this experiment to Section~\ref{sec:experiments}.

Recall that the user identification problem, given \emph{only} the
graph information, is intractable in general graphs. Even the special
case where $s_1 = s_2 = 1$ (that is, no edges have been deleted) is
equivalent to the well-studied Graph Isomorphism problem, for which no
polynomial-time algorithm is known. Fortunately, in reality, there are
additional sources of information which allow one to identify a subset
of nodes across the two networks: For example, people can use the same
email address to sign up on multiple websites. Users often explicitly
connect their network accounts, for instance by posting a link to
their Facebook profile page on Google+ or Twitter and vice versa. To
model this, we assume that there is a set of users/nodes explicitly
linked across the two networks $G_1, G_2$. More formally, there is a
\emph{linking probability} $l$ (typically, $l$ is a small constant)
and each node in $V$ is linked across the networks independently with
probability $l$. (In real networks, nodes may be linked with differing
probabilities, but high-degree nodes / celebrities may be more likely
to connect their accounts and engage in cross-network promotions; this
would be more likely to help our algorithm, since low-degree nodes
are less valuable as seeds because they help identify only a small
number of neighbors. In the experiments of \cite{deanonymizing09}, the
authors explicitly consider high-degree nodes as seeds in the
real-world experiments.)

In Section~\ref{subsec:algo} below, we present a natural algorithm to
solve the user identification problem with a set of linked nodes, and
discuss some of its properties. Then, in Section~\ref{sec:theory}, we
prove that this algorithm performs well on several well-established
network models. In Section~\ref{sec:experiments}, we show that the
algorithm also works very well in practice, by examining its
performance on real-world networks.

\subsection{The Algorithm}
\label{subsec:algo}

To solve the user identification problem, we design a local
distributed algorithm that uses only structural information about the
graphs to expand the initial set of links into a mapping/identification
of a large fraction of the nodes in the two networks.

Before describing the algorithm, we introduce a useful definition.

\begin{definition}
A pair of nodes $(u_1, u_2)$ with $u_1 \in G_1, u_2 \in G_2$ is said
to be a \emph{similarity witness} for a pair $(v_1, v_2)$ with $v_1 \in
G_1, v_2 \in G_2$ if $u_1 \in N_1(v_1), u_2 \in N_2(v_2)$ and $u_1$
has been linked to / identified with $u_2$.
\end{definition}

Here, $N_1(v_1)$ denotes the neighborhood of $v_1$ in $G_1$, and
similarly $N_2(v_2)$ denotes the neighborhood of $v_2$ in $G_2$.

Roughly speaking, in each phase of the algorithm, every pair of nodes
(one from each network) computes an similarity score that is equal to
the number of similarity witnesses they have. We then create a link
between two nodes $v_1$ and $v_2$ if $v_2$ is the node in $G_2$ with
maximum similarity score to $v_1$ and vice versa. We then use the
newly generated set of links as input to the next phase of the
algorithm.

A possible risk of this algorithm is that in early phases, when few
nodes in the network have been linked, low-degree nodes could be
mis-matched. To avoid this (improving precision), in the $i$th phase,
we only allow nodes of degree roughly $D/2^i$ and above to be matched,
where $D$ is a parameter related to the largest node degree. Thus, in the
first phase, we match only the nodes of very high degree, and in
subsequent phases, we gradually decrease the degree threshold required
for matching. In the experimental section we will show in fact that this
simple step is very effective, reducing the error rate by more than $33\%$.
We summarize the algorithm, that we called User-Matching, as follows:

{\footnotesize
\begin{algo}
  \bf{Input:} \\$G_1(V,E_1), G_2(V,E_2), L$ a set of initial identification links\\
  across the networks, $D$ the maximum degree in the graph\\
  a minimum matching score $T$ and a specified number of\\
  iteration $k$.\\
  \bf{Output}: \\
  A larger set of identification links across the networks.\\
  \bf{Algorithm}: \\
  For $i=1,\dots, k$ \+ \\
  For $j=\log D,\dots, 1$\+ \\
  For all the pairs $(u,v)$ with $u\in G_1$ and $v\in G_2$\\
  and such that $d_{G_1}(u)\geq 2^j$ and $d_{G_2}(v)\geq 2^j$\+ \\
  Assign to $(u,v)$ a score equal to the number\\
  of similarity witnesses between $u$ and $v$\- \\
  If $(u,v)$ is the pair with highest score in which\\
  either $u$ or $v$ appear and the score is above $T$ \\
  add $(u,v)$ to $L$.\-\- \\
  Output $L$
\end{algo}
}

Where $d_{G_i}(u)$ is the degree of node $u$ in $G_i$.  Note that the
internal for loop can be implemented efficiently with $4$ consecutive
rounds of MapReduce, so the total running time would consist of
$O(k\log D)$ MapReductions. In the experiments, we note that even for
a small constant $k$ ($1$ or $2$), the algorithm returns very
interesting results. The optimal choice of threshold $T$ depends on
the desired precision/recall tradeoff; higher choices of T improve
precision, but in our experiments, we note that $T = 2$ or $3$ is
sufficient for very high precision.

\begin{figure}
  \begin{center}
    \includegraphics[scale=0.30]{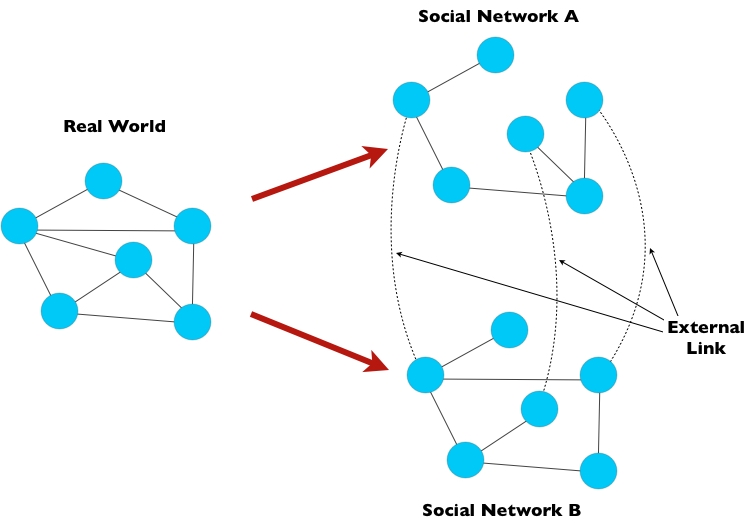}
  \end{center}
  %\vspace{-0.2in}
  \caption{From the real underlying social network, the model
    generates two random realizations of it, A and B, and some
    identification links for a subset of the users across the two
    realizations} \label{fig:model}
\end{figure}

\section{Theoretical Results}
\label{sec:theory}

In this section we formally analyze the performance of our algorithm
on two network models. In particular, we explain why our simple
algorithm should be effective on real social networks. The core idea
of the proofs and the algorithm is to leverage high degree nodes to
discover all the possible mapping between users. In fact, as we show
here theoretically and later experimentally, high degree nodes are
easy to detect. Once we are able to detect the high degree nodes, most
low degree nodes can be identified using this information.

We start with the \er (Random) Graph model~\cite{er59} to warm up with
the proofs, and to explore the intuition behind the algorithm. Then we
move to our main theoretical results, for the Preferential Attachment
Model.  For simplicity of exposition, we assume throughout this
section that $s_1 = s_2 = s$; this does not change the proofs in any
material detail.

\subsection{Warm up: Random Graphs}
\label{subsec:Gnp}

In this section, we prove that if the underlying `true' network is a
random graph generated from the \er model (also known as $G(n, p)$),
our algorithm identifies almost all nodes in the network with high
probability.

Formally, in the $G(n, p)$ model, we start with an empty graph on $n$
nodes, and insert each of the $\binom{n}{2}$ possible edges
independently with probability $p$. We assume that $p < 1/6$; in fact,
any constant probability results in graphs which are much denser than
any social network.\footnote{In fact, the proof works even with
  $p=1/2$, but it requires more care. However, when $p$ is too close
  to 1, $G$ is close to a clique and all nodes have near-identical
  neighborhoods, making it impossible to distinguish between them.}
Let $G$ be a graph generated by this process; given this underlying
network $G$, we now construct two partial realizations $G_1, G_2$ as
described by our model of Section~\ref{sec:model}.

We note that the probability a specific edge exists in $G_1$ or $G_2$
is $p s$. Also, if $n p s$ is less than $(1-\eps) \log n$ for $\eps >
0$, the graphs $G_1$ and $G_2$ are not connected
w.h.p. \cite{er59}. Therefore, we assume that $n p s > c\log n$ for some
constant $c$.

In the following we identify the nodes in $G_1$ with $u_1,\dots,u_n$
and the nodes in $G_2$ with $v_1,\dots,v_n$, where nodes $u_i$ and
$v_i$ correspond to the same node $i$ in $G$.  In the first phase, the
expected number of similarity witnesses for a pair $(u_i, v_i)$ is
$(n-1) p s^2 \cdot l$. This follows because the expected number of
neighbors of $i$ in $G$ is $(n-1)p$, the probability that the edge to
a given neighbor survives in both $G_1$ and $G_2$ is $s^2$, and the
probability that it is initially linked is $l$. On the other hand, the
expected number of similarity witnesses for a pair $(u_i, v_j)$, with
$i\neq j$ is $(n-2) p^2 s^2 \cdot l$; the additional factor of $p$ is
because a given other node must have an edge to both $i$ and $j$,
 which occurs with probability $p^2$. Thus, there is a factor of
$p$ difference between the expected number of similarity witnesses a
node $u_i$ has with its true match $v_i$ and with some other node
$v_j$, with $i\neq j$. The main intuition is that this factor of $p <
1$ is enough to ensure the correctness of algorithm. We prove this by
separately considering two cases: If $p$ is sufficiently large, the
expected number of similarity witnesses is large, and we can apply a
concentration bound. On the other hand, if $p$ is small, $n p^2 s^2$
is so small that the expected number of similarity witnesses is almost
negligible.

We start by proving that in the first case there is always a gap between
a real and a false match.

\begin{theorem}\label{thm:p-large}
  If $(n-2) p s^2 l \ge 24 \log n$ (that is, $p > \frac{24}{s^2l}
  \frac{\log n}{n-2}$), w.h.p. the number of first-phase similarity
  witnesses between $u_i$ and $v_i$ is at least $(n-1)ps^2l/2$. The number of first-phase
  similarity witnesses between $u_i$ and $v_j$, with $i\neq j$ is w.h.p. at most $(n-2)ps^2l/2$.
\end{theorem}
\begin{proof}
  We prove both parts of the lemma using Chernoff Bounds (see, for
  instance, \cite{concentration}). 

  Let consider a pair for node $j$. Let $Y_i$ be a r.v. such that
  $Y_i = 1$ if node $u_i\in N_1(u_j)$ and $v_i\in N_2(v_j)$, and if
  $(u_i,v_i)\in L$, where $L$ is the initial seed of links across $G_1$ and
  $G_2$. Then, we have $Pr[Y_1 = 1]
  = p s^2 l$. If $Y = \sum_{i=1}^{n-1}Y_i$, the Chernoff bound implies
  that $Pr[Y < (1-\delta)E[Y]] \le e^{-E[Y]\delta^2/2}$. That is,
  \[Pr \left[Y < \frac{1}{2} (n-1)ps^2l \right] \le e^{-E[Y] / 8} < e^{-3\log n} = 1/n^3\]
  Now, taking the union bound over the $n$ nodes in $G$, w.h.p. every node
  has the desired number of first-phase similarity witnesses with its
  copy.

  To prove the second part, suppose w.l.o.g. that we are considering
  the number of first-phase similarity witnesses between $u_i$ and 
  $v_j$, with $i\neq j$. Let $Y_i = 1$ if node $u_z\in  N_1(u_i)$ and 
  $v_z\in N_2(v_j)$, and if $(u_z,v_x)\in L$. If $Y =
  \sum_{i=1}^{n-2}Y_i$, the Chernoff bound implies that $Pr[Y >
    (1+\delta) E[Y]] \le e^{-E[Y] \delta^2/4}$. That is,
  \[ Pr \left[Y > \frac{1}{2p} (n-2)p^2 s^2 l = \frac{(n-2)ps^2l}{2} \right]
  \le e^{-E[Y] ({1 \over 2p} - 1)^2 / 4}\]
  \[ = e^{-2p ({1 \over 2p} -1)^2 3\log n} \le 1/n^3 \]
  where the last inequality comes from the fact that $p < 1/6$. Taking
  the union bound over all $n(n-1)$ unordered pairs of nodes $u_i,
  v_j$ gives the fact that w.h.p., every pair of different nodes does
  not have too many similarity witnesses. 
\end{proof}

The theorem above implies that when $p$ is sufficiently large, there is a
gap between the number of similarity witnesses of pairs of nodes that correspond
to the same node and a pair of nodes that do not correspond to the same node. 
Thus the first-phase similarity witnesses are enough to completely distinguish
between the correct copy of a node and possible incorrect matches. 

It
remains only to consider the case when $p$ is smaller than the bound
required for Theorem~\ref{thm:p-large}. This requires the following
useful lemma.

\begin{lemma}\label{lem:at-most-c}
  Let $B$ be a Bernoulli random variable, which is 1 with probability
  at most $x$, and 0 otherwise. In $k$ independent trials, let $B_i$ 
  denote the outcome of the $i$\emph{th} trial,
  and let $B(k) = \sum_{i=1}^k B_i$:
  %$\begin{enumerate}
  %\item
  If $kx$
  is $o(1)$, the probability that $B(k)$ is greater than $2$ is at most
  $k^3x^3/6 + o(k^3x^3)$.
  %\item The probability that $B(k)$ is greater than $6$ is at most $o()$
  %\end{enumerate}
\end{lemma}
\begin{proof}
  The probability that $B(k)$ is at most 2 is given by:
  $(1-x)^k + kx \cdot (1-x)^{k-1} + \binom{k}{2} x^2 \cdot 
  (1-x)^{k-2}$. Using the Taylor series expansion for $(1-x)^{k-2}$,
  this is at most $1 - k^3x^3/6 - o(k^3x^3)$.
\iffalse 
 \[(1-x)^k + kx \cdot (1-x)^{k-1} + \binom{k}{2} x^2 \cdot 
  (1-x)^{k-2}\]
  \[ = (1-x)^{k-2} \left[(1-x)^2 + kx(1-x) + \binom{k}{2}x^2 \right]\]
  \[ \ge \left(1 - (k-2)x + \binom{k-2}{2}x^2 -
  \binom{k-2}{3}x^3\right) \cdot \]
  \[\left[(1-x)^2 + kx(1-x) + \binom{k}{2}x^2 \right]\]
  \[ = 1 - k^3x^3/6 - o(k^3x^3)\]
\fi
\end{proof}

When we run our algorithm on a graph drawn from $G(n, p)$, we set the
minimum matching threshold to be $3$.

\begin{lemma}
  If $p \le \frac{24}{s^2l} \frac{\log n}{n-2}$, w.h.p., algorithm \um
  never incorrectly matches nodes $u_i$ and $v_j$ with $i\neq j$.
\end{lemma}
\begin{proof}
  Suppose for contradiction the algorithm does incorrectly match two
  such nodes, and consider the first time this occurs. We use
  Lemma~\ref{lem:at-most-c} above. Let $B_z$ denote the event that the
   vertex $z$ is a similarity witness for $u_i$ and
  $v_j$. 

  In order for $B_z$ to have occurred, we must have $u_z$ in
  $N_1(U_i)$ and $v_z$ in $N_2(v_j)$ and $(u_z,v_z)\in L$. 
  The probability that $B_z = 1$ is therefore at
  most $p^2 s^2$. Note that each $B_z$ is independent of the others,
  and that there are $n-2$ such events. As $p$ is $O(\log n/n)$, the
  conditions of Lemma~\ref{lem:at-most-c} apply, and hence the
  probability that more than 2 such events occur is at most $(n-2)^3
  p^6 s^6$.  But $p$ is $O(\log n/n)$, and hence this event occurs
  with probability at most $O(\log^6n / n^3)$. Taking the union bound
  over all $n(n-1)$ unordered pairs of nodes $u_i, v_j$ gives the fact
  that w.h.p., not more than 2 similarity witnesses can exist for any
  such pair. But since the minimum matching threshold for our algorithm is $3$,
  the algorithm does not incorrectly match this pair, contradicting
  our original assumption.
\end{proof}

Having proved that our algorithm does not make errors, we now show
that it identifies most of the graph.

\begin{theorem}
  Our algorithm identifies $1 - o(1)$ fraction
  of the nodes w.h.p.
\end{theorem}
\begin{proof}
  Note that the probability that a node is identified is $1-o(1)$ by
  the Chernoff bound because in expectation it has $\Omega(\log n)$
  similarity witnesses. So in expectation, we identify $1 - o(1)$
  fraction of the nodes. Furthermore, by applying the method of
  bounded difference~\cite{concentration} (each node affects the final
  result at most by $1$), we get that the result holds also with high
  probability.
\end{proof}

\subsection{Preferential Attachment}
\label{subsec:pref}

The preferential attachment model was introduced by Barab\'{a}si and Albert 
in~\cite{ba99}. In this paper we consider the formal definition of the
model described in~\cite{br}.
\begin{definition}{\bf [PA model].}\label{def:PA}
Let $G_n^m$, $m$ being a fixed parameter, be defined inductively as follows:
\begin{itemize}
\item
  $G_1^m$ consists of a single vertex with $m$ self-loops.
\item
  $G_n^m$ is built from $G_{n-1}^m$ by adding a new node $u$ together
  with $m$ edges $e_u^1=(u,v_1), \ldots, e_u^m=(u,v_m)$ inserted one
  after the other in this order. Let $M_i$ be the sum of the degrees of all
  the nodes when the edge $e_u^i$ is added. The endpoint $v_i$ is
  selected with probability $\frac{deg(v_i)}{M_i+1}$, with the exception 
  of node $u$, which is selected with probability $\frac{d(u)+1}{M_i+1}$.
  %instead of $\frac{i-1}{M_i+1}$.
\end{itemize}
\end{definition}

The PA model is the most celebrated model for social networks. Unlike
the \er model, in which all nodes have roughly the same degree, PA
graphs have a degree distribution that more accurately resembles the
skew degree distribution seen in real social networks.  Though more
evolved models of social networks have been recently introduced, we
focus on the PA model here because it clearly illustrates why our
algorithm works in practice. Note that the power-law distribution of
the model complicates our proofs, as the overwhelming majority of
nodes only have constant degree ($\le 2m$), and so we can no longer
simply apply concentration bounds to obtain results that hold
w.h.p. For a (small) constant fraction of the nodes $u$, there does
not exist any node $z$ such that $u_z \in N_1(u_i)$ and $v_z \in
N_2(v_i)$; we cannot hope to identify these nodes, as they have no
neighbors ``in common'' on the two networks. In fact, if $m = 4$ and
$s=1/2$, roughly $30\%$ of nodes of ``true'' degree $m$ will be in
this situation. Therefore, to be able to identify a reasonable
fraction of the nodes, one needs $m$ to be at least a reasonably large
constant; this is not a serious constraint, as the median friend count
on Facebook, for instance, is over 100.  In our experimental section,
we show that our algorithm is effective even with smaller $m$.

We now outline our overall approach to identify nodes across two PA
graphs. In Lemma~\ref{lem:PA-large-degree}, we argue that for the
nodes of very high degree, their neighborhoods are different enough
that we can apply concentration arguments and uniquely identify
them. For nodes of intermediate degree ($\log^3 n$) and less, we argue
in Lemma~\ref{lem:PA-few-neighbors} that two distinct nodes of such
degree are very unlikely to have more than 8 neighbors in
common. Thus, running our algorithm with a minimum matching threshold
of 9 guarantees that there are no mistakes.  Finally, we prove in
Lemma~\ref{lem:PA-iterations} that when we run the algorithm
iteratively, the high-degree nodes help us identify many other nodes,
these nodes together with the high-degree nodes in turn help us
identify more, and so on: Eventually, the probability that any given
node is unidentified is less than a small constant, which implies that
we correctly identify a large fraction of the nodes.

Interestingly, we notice in our experiments that on real networks, the
algorithm has the same behavior as on PA graphs. In fact, as we will
discuss later, the algorithm is always able to identify
high-degree/important nodes and then, using this information, identify
the rest of the graph.

\noindent \textbf{Technical Results:}
The first of the three main lemmas we need,
Lemma~\ref{lem:PA-large-degree}, states that we can identify all
of the high-degree nodes correctly. To prove this, we need a few
technical results.
%due to space constraints, we omit the proofs from this version. 
These results say that all nodes of high degree join the
network early, and continue to increase their degree significantly
throughout the process; this helps us show that high-degree nodes do
not share too many neighbors.

\subsubsection{High degree nodes are early-birds}\label{subsec:early-birds}

Here we will prove formally that the nodes of degree $\Omega (\log^2
n)$ join the network very early; this will be useful to show that two
high degree nodes do not share too many neighbors.

\begin{lemma}\label{early-birds}
Let $G_{n}^{m}$ be the preferential attachment graph obtained after
$n$ steps. Then for any node $v$ inserted after time $\psi n$, for any
constant $\psi>0$, $d_{n}(v)\in o(log^2 n)$ with high probability,
where $d_{n}(v)$ is the degree of nodes $v$ at time $n$.
\end{lemma}

\begin{proof}
It is possible to prove that such nodes have expected constant degree,
but unfortunately, it is not trivial to get a high probability result
from this observation because of the inherent dependencies that are
present in the preferential attachment model. For this reason we will
not prove the statement directly, but we will take a short detour
inspired by the proof in~\cite{thesis}. In particular we will first break
the interval in a constant number of small intervals. Then we will
show that in each interval the degree of $v$ will increase by at most
$O(\log n)$ with high probability. Thus we will be able to conclude
that at the end of the process the total degree of $v$ is at most
$O(\log n)$(recall that we only have a constant number of interval).

As mentioned above we analyze the evolution of the degree of $v$ in
the interval $\psi n$ to $n$ by splitting this interval in a constant
number of segments of length $\lambda n$, for some constant $\lambda >
0$ to be fixed later. Now we can focus on what happens to the degree
of $v$ in the interval $(t, \cdot \lambda n + t]$ if $d_{t}(v) \leq C
  \log n$, for some constant $C \geq 1$ and $t \geq \psi n$. Note that
  if we can prove that $d_{\lambda n + t} \leq C' \log n$, for some constant
  $C'\geq 0$ with probability $1 - o\left(n^{-2}\right)$, we can then
  finish the proof by the arguments presented in the previous
  paragraph.

In order to prove this, we will take a small detour to avoid the
dependencies in the preferential attachment model. More specifically,
we will first show that this is true with high probability for a
random variable $X$ for which it is easy to get the concentration
result. Then we will prove that the random variable $X$ stochastically
dominates the increase in the degree of $v$. Thus the result will
follow.

Now, let us define $X$ as the number of heads that we get when we toss
a coin which gives head with probability $\frac{C' \log n}{t}$ for
$\lambda n$ times, for some constant $C'\geq 13 C$. It is possible to
see that:
$$E[X] = \frac{C' \log n}{t}\lambda n \leq \frac{C' \log n}{\psi
  n}\lambda n \leq \frac{C' \lambda\log n}{\psi}$$
 Now we fix $\lambda = \frac{\psi}{100}$ and we use the Chernoff bound
 to get the following result:
\begin{eqnarray*}
\Pr\left(X > \frac{C'-C}{2} \log n\right) & = & \Pr\Bigg(X > \left(\frac{100(C'- C)}{2C'}\right)E[X]\Bigg)\\
& \leq & 2^{ -\frac{C'}{100}\log n\left(\frac{100(C'- C)}{2C'}\right)} \\
& \leq & 2^{ -\left(\frac{(6C')}{13}\log n\right)}\\
& \leq & 2^{ -6\log n} \in O\left(n^{-3}\right) 
\end{eqnarray*}

So we know that the value of $X$ is bounded by $\frac{C'-C}{2} \log n$
with probability $O\left(n^{-3}\right)$. Now, note that until the
degree of $v$ is less than $C' \log n$ the probability that $v$
increases its degree in the next step is stochastically dominated by
the probability that we get an head when we toss a coin which gives
head with probability $\frac{C' \log n}{t}$. To conclude our algorithm
we study the probability that $v$ become of degree $\frac{C' \log
  n}{t}$ precisely at time $t \leq t^* \leq \lambda n$. Note that until
time $t^*$ $v$ has degree smaller than $\frac{C' \log n}{t}$ and so it
is dominated by the coin. But we already know that when we toss such a
coin at most $\lambda n$ times the probability of getting
$\frac{C'-C}{2} \log n$ heads is in $O\left(n^{-3}\right)$.  Thus for
any $t \leq t^* \leq \lambda n$ the probability that $v$ reach degree
$C' \log n$ at time $t^*$ is $O(n^{-3})$. Thus by using the union bound
on all the possible $t^*$, $v$ will get to degree $C' \log n$ with
probability $O(n^{-2})$.

At this point we can finish the proof by taking the union bound on all
the segments(recall that they are constant) $(\psi n, \psi + \lambda
n], (\psi + \lambda n, \psi n + 2\lambda], \cdots$ and on the number
    of nodes and we get that all the nodes that join the network after
    time $\psi n$ have degree that is upper bounded by $C'' \log n$
    for some constant $C'' \geq 0$ with probability $O(n^{-1})$.
\end{proof}

\subsubsection{The rich get richer}

In this section we study another fundamental property of the
preferential attachment, which is that nodes with degree bigger than
$\log^2 n$ continue to increase their degree significantly throughout
the process. More formally:

\begin{lemma}\label{lem:rich-get-richer}
  Let $G_{n}^{m}$ be the preferential attachment graph obtained after
  $n$ steps. Then with high probability for any node $v$ of degree $d
  \geq \log^2 n$ and for any fixed constant $\epsilon \geq 0$, a
  $\frac{1}{3}$ fraction of the neighbors of $v$ joined the graph
  after time $\epsilon n$.
\end{lemma}
\begin{proof}
  By Lemma~\ref{early-birds} above, we know that $v$ joined the
  network before time $\epsilon n$ for any fixed constant $\epsilon
  \geq 0$. Now we consider two cases. In the first, $d_{\epsilon n}(v)
  \leq \frac12 \log^2 n$, in which case the statement is true because
  the final degree is bigger than $\log^2 n$.  Otherwise, we have that
  $d_{\epsilon n}(v) > \frac12 \log^2 n$, in this case the probability
  that $v$ increase its degree at every time step after $\epsilon n$
  dominates the probability that a toss of a biased coin which gives
  head with probability $\frac{\log^2 n}{2mn}$ comes up head.  Now
  consider the random variable $X$ that counts the number of heads
  when we toss a coin that lands head with probability $\frac{\log^2
    n}{2mn}$ for $(1-\epsilon) mn$ times. The expected value of $X$
  is:
  $$E[X] = \frac{\log^2 n}{2mn}(1-\epsilon) mn = \frac{1-\epsilon}{2}
  \log^2 n $$ Thus using the Chernoff bound:
  \begin{eqnarray*}
    \Pr\left(X \leq \frac{1-2\epsilon}{2} \log^2 n\right)& \leq & exp\left(-\frac12\left(1-\frac{\epsilon}
            {1-\epsilon}\right)\log^2 n\right) \\
            &\in& O(n^{-2})
  \end{eqnarray*}

  Thus with probability $O(1-n^{-2})$ $X$ is bigger that
  $\frac{1-2\epsilon}{2} \log^2 n$ but as mentioned before the
  increase in the degree of $v$ stochastically dominates $X$. Thus
  taking the union bound on all the possible $v$ we get that the
  statement holds with probability equal to $O(1-n^{-1})$. Thus the
  claim follows.
\end{proof}

\subsubsection{First-mover advantage}
\begin{lemma}\label{lem:first-mover}
  Let $G_{n}^{m}$ be the preferential attachment graph obtained after
  $n$ steps. Then with high probability all the nodes that join the
  network before time $n^{0.3}$ have degree at least $\log^3 n$.
\end{lemma}

\begin{proof}
  To prove this theorem we will use some results from~\cite{cf07}, but
  before we need to introduce another model equivalent to the
  preferential attachment. In this new process instead of constructing
  $G_{n}^{m}$, we first construct $G_{nm}^{1}$ and then we collapse
  the vertices $1,\cdots, m$ to construct the first vertex, the vertex
  between $m+1, \cdots 2m$ to construct the second vertex and so on so
  for. It is not hard to see that this new model is equivalent to the
  preferential attachment. Now we can prove our technical theorem.

  Now we can state two useful equation from the proof of Lemma 6
  in~\cite{cf07}. Consider the model $G_{1}^{nm}$. Let
  $D_k=d_{nm}(v_1)+d_{nm}(v_2)+\cdots +d_{nm}(v_k)$, where
  $d_{nm}(v_i)$ is the degree of a node inserted at time $i$ at time
  $nm$. Then $k\geq 1$ we have:
  \begin{eqnarray}\label{cooper_first}
    \Pr\left(|D_k -2\sqrt{kmn}|\geq 3\sqrt{mn \log(mn)}\right)\leq (mn)^{-2}
  \end{eqnarray}
  From the same paper we also have that if $0\leq d \leq mn -k - s$, we can derive from
  equation (23) that
  \begin{eqnarray}
    \Pr(d_n(v_{k+1}) = d+1 | D_k -2k = s) \leq \frac{s+d}{2N-2k-s-d}
  \end{eqnarray}

  From~\ref{cooper_first} we can derive that:
  $$\Pr\left(D_k-2k\geq 3\sqrt{mn\log(mn)} + 2\sqrt{kmn}
  -2k\right)\leq (mn)^{-2}$$
  $$\Pr\left(D_k-2k\geq 5\sqrt{kmn\log(mn)} \right)\leq (mn)^{-2}$$
  Thus we get that:
  \begin{eqnarray*}
    &&\Pr(d_n(v_{k+1}) < \log^3 n)  =   \sum_{0}^{\log^3 n - 1}\Pr(d_n(v_{k+1}) = i)\\
    &\leq&  \Pr\left(D_k-2k\geq 3\sqrt{mn\log(mn)} + 2\sqrt{kmn} -2k\right)\\
    &+& \sum_{i=0}^{\log^3 n - 1}\sum_{j=0}^{5\sqrt{kmn\log(mn)}}\Big(\Pr\left(D_k-2k = j\right)\\
    &&  \Pr(d_n(v_{k+1}) = i | D_k -2k = j)\Big) \\
    &\leq&  (mn)^{-2} +\sum_{i=0}^{\log^3 n - 1} \Pr\Big(d_n(v_{k+1}) = i | D_k -2k \\
    &&= 5\sqrt{mn\log(mn)}\Big) \\
    &\leq& (mn)^{-2} +\sum_{i=0}^{\log^3 n - 1}\frac{5\sqrt{mn\log(mn)}+i-1}{2mn-2k-5\sqrt{mn\log(mn)}-i+1}\\
    &\in& O\left(\frac{\log^4(n)}{\sqrt{n}}\right)
  \end{eqnarray*}
  where we assumed that $k\in O\left(n^{\frac{1}{3}}\right)$.
  
  So now by union bounding on the first $mn^{0.3}$ nodes we obtain that
  with high probability in $G_{1}^{nm}$ all the nodes have degree bigger than
  $\log^2 n$. But this implies in turn the statement of the theorem by construction of
  $G_{1}^{nm}$ .
\end{proof}

\subsubsection{Handling product of generalized harmonic}
\begin{lemma}
  Let $a$ and $b$ be constant greater than $0$. Then:
  $$\sum_{i=n^a}^{n^b-2}\sum_{j>i}^{n^b-1}\sum_{z>j}^{n^b}\frac{1}{i^2j^2z^2}\in
  O\left(\frac{1}{n^{3a}}\right)$$
\end{lemma}
\begin{proof}
  Recall that $\sum_{k=x}^{y}\frac{1}{k^2} \in
  \Theta(\frac{1}{x}-\frac{1}{y})$. Then
  \begin{eqnarray*}
    \sum_{i=n^a}^{n^b-2}\sum_{j>i}^{n^b-1}\sum_{z>j}^{n^b}\frac{1}{i^2j^2z^2} &=& \sum_{i=n^a}^{n^b-2}\frac{1}{i^2}\sum_{j>i}^{n^b-1}\frac{1}{j^2}\sum_{z>j}^{n^b}\frac{1}{z^2}\\
    &\in & \Theta(\frac{1}{n^a}-\frac{1}{n^b})\Theta(\frac{1}{n^a}-\frac{1}{n^b})\Theta(\frac{1}{n^a}-\frac{1}{n^b})\\
    &\in &  O\left(\frac{1}{n^{3a}}\right)
  \end{eqnarray*}
  \end{proof}

\noindent \textbf{Completing the Proof:}
We now use the technical lemmas above to complete our proof for the
preferential attachment model.

\begin{lemma}\label{lem:arrival-prob} 
  For a node $u$ with degree $d$, the probability that it is incident
  to a node arriving at time $i$ is at most $\max\{d, \log^3 n\} /
  (i-1)$ w.h.p.
\end{lemma}
\begin{proof}
  If node $i$ arrives after $u$, the probability that $i$ is adjacent
  to $u$ is at most the given value, since there are $m(i-1)$ edges
  existing in the graph already, and we take the union bound over the
  $m$ edges incident to $i$. If $i$ arrives \emph{before} $u$, let $t$
  denote the time at which $u$ arrives. From Lemma~6 of \cite{cf07},
  the degree of $i$ at t is at most $\sqrt{t}{i} \log^3 n$ w.h.p.. But there
  are $(t-1)m$ edges already in the graph at this time, and since $u$
  has $m$ edges incident to it, the probability that one of them is
  incident to $i$ is at most $\log^3 n \frac{\sqrt{t}}{\sqrt{i}(t-1)}
  \le \log^3 n / (i-1)$.
\end{proof}

\begin{lemma}\label{lem:PA-few-neighbors}
   W.h.p, for any pair of nodes $u, v$ of degree $< \log^3 n$, $|N(u)
   \cap N(v)| \le 8$.
\end{lemma}
\begin{proof}
  From Lemma~\ref{lem:first-mover}, nodes $u$ and $v$ must have
  arrived after time $t = n^{0.3}$. Let $a, b$ be constants such that
  $0.3 < a < b < 1$ and $b \le 3/2 a - \eps$ for some constant $\eps >
  0$. We first show that the probability that any two nodes $u, v$
  with degree less than $\log^3 n$ and arriving before time $n^b$ have
  $3$ or more common neighbors between $n^a$ and $n^b$ is at most
  $n^{-\eps}$. This implies that, setting $a$ to $0.3$, nodes $u$ and
  $v$ have at most 2 neighbors between $n^a$ and $n^{3a/2 - \eps}$, at
  most 2 between $n^{3a/2 - \eps}$ and $n^{9a/4}$, and at most 2
  between $n^{9a/4}$ and $n^{27a/8} > n$, for a total of 6
  overall. Similarly, we show that $u$ and $v$ have at most 2
  neighbors arriving before $n^{0.3}$, which completes the lemma.
  
  From Lemma~\ref{lem:arrival-prob} above, the probability that a node
  arriving at time $i$ is incident to $u$ and $v$ is at most $(\log^3
  n / (i-1))^2$. (The events are not independent, but they are
  negatively correlated.) The probability that 3 nodes $i, j, k$ are
  all incident to both $u$ and $v$, then, is at most $(\log^3 n)^6 /
  ((i-1)(j-1)(k-1))^2$. Therefore, for a fixed $u, v$, the probability
  that \emph{some} 3 nodes are adjacent to $u$ and $v$ is at most:
  \begin{eqnarray*} 
    &&\log^{18} n \sum_{i=n^a}^{n^b} \sum_{j=n^a}^{n^b}
    \sum_{k=n^a}^{n^b} \frac{1}{\left((i-1)(j-1)(k-1)\right)^2} \\
    &\le& \log^{18} n \left(\frac{1}{n^a} - \frac{1}{n^b} \right)^3
  \end{eqnarray*}

  There are at most $n^b$ choices for each of $u$ and $v$; taking the
  union bound, the probability that any pair $u$, $v$ have 3 or more
  neighbors in common is at most $n^{2b - 3a} \log^{18} n = n^{-2 \eps}
  \log^{18}n$.
\end{proof}

So, by setting the matching threshold to $9$, the algorithm never
makes errors; we now prove that it actually detects a lot of good new
links.

\begin{lemma}\label{lem:PA-large-degree}
  The algorithm successfully identifies any node of degree $\ge 4
  \log^2n / (s^2 l)$.
\end{lemma}
\begin{proof}
  For any node $v$ of degree $d(v) \ge 4 \log^2 n / (s^2 l)$, the
  expected number of similarity witnesses it has with its copy during
  the first phase is $d(v) s^2 l$; using the Chernoff Bound, the
  probability that the number is less than $7/8$ of its expectation is
  at most $\exp(-d(v) s^2 l / 128) \le \exp(-\log^2 n / 32) =
  \frac{1}{n^{\log n / 32}}$. Therefore, with very high probability,
  every node $v$ of degree $d(v)$ has at least $7/8 \cdot d(v) s^2 l$
  first-phase similarity witnesses with its copy.

  On the other hand, how many similarity witnesses can node $v$ have
  with a copy of a different node $u$? Fix $\eps > 0$, and first
  consider potential similarity witnesses that arrive before time $t =
  \eps n$; later, we consider those that arrive after this time. From
  Lemma~\ref{lem:rich-get-richer}, we have $d_t(v) \le (2/3 + \eps)
  d(v)$. Even if \emph{all} of these neighbors of $v$ are also
  incident to $u$, the expected number of similarity witnesses for
  $(u, v)$ is at most $d_t(v) s^2 l$. Now consider the neighbors
  of $v$ that arrive after time $\eps n$. Each of these nodes is a
  neighbor of $u$ with probability $\le d(u) / (2 m \eps n)$. But
  $d(u) \le \tilde{O}(\sqrt{n})$, and hence each of the neighbors of
  $v$ is a neighbor of $u$ with probability $o(1/n^{1/2 -
    \delta})$. Therefore, the expected number of similarity witnesses
  for $(u, v)$ among these nodes is at most $d(v) s^2 l / n^{1/2 -
    \delta}$. Therefore, the total expected number of similarity
  witnesses is at most $(2/3 + \eps) d(v) s^2 l$. Again using the
  Chernoff Bound, the probability that this is at least $7/8 \cdot
  d(v) s^2 l$ is at most $\exp(-3/4 d(v)s^2 l \cdot / 64) =
  \exp(-3\log^2 n / 64)$, which is at most $\frac{1}{n^{3\log n /
      64}}$. 

  To conclude, we showed that with very high probability, a
  high-degree node $v$ has at least $7/8 \cdot d(v) s^2 l$ first-phase
  similarity witnesses with its copy, and has fewer than this number
  of witnesses to the copy of any other node. Therefore, our algorithm
  correctly identifies all high-degree nodes.
\end{proof}

From the two preceding lemmas, we identify all the high-degree nodes,
and make no mistakes on the low-degree nodes. It therefore remains
only to show that we have a good success probability for the
low-degree nodes. In the lemma below, we show this when $ms^2 \ge
22$. We note that one still obtains good results even with a higher or
lower value of $ms^2$, but it appears difficult to obtain a simple
closed-form expression for the fraction of identified nodes.  For ease
of exposition, we present the case of $ms^2 \ge 22$ here, but the
proof generalizes in the obvious way.

\begin{lemma}\label{lem:PA-iterations}
  Suppose $ms^2 \ge 22$. Then, w.h.p., we successfully identify at
  least $97\%$ of the nodes.
\end{lemma}
\begin{proof}
  We have already seen that all high-degree nodes (those arriving
  before time $n^{0.3}$) are identified in the first phase of the
  algorithm. Note also that it never makes a mistake; it
  therefore remains only to identify the lower-degree nodes.  We
  describe a sequence of iterations in which we bound the probability
  of failing to identify other nodes. 

  Consider a sequence of roughly $n^{0.75}$ iterations, in each of
  which we analyze $n^{0.25}$ nodes. In particular, iteration $i$
  contains all nodes that arrived between time $n^{0.3} +
  (i-1)n^{0.25}$ and time $n^{0.3} + i \cdot n^{0.25}$. We argue
  inductively that after iteration $i$, w.h.p. the fraction of nodes
  belonging to this iteration that are not identified is less than
  $0.03$, and the total fraction of degree incident to unidentified
  nodes is less than $0.08$. Since this is true for each $i$, we obtain
  the lemma.

  The base case of nodes arriving before $n^{0.3}$ has already been
  handled. Now, note that during any iteration, the total degree
  incident to nodes of this iteration is at most $2m n^{0.25} \ll
  n^{0.3}$. Thus, when each node of this iteration, the probability
  that any of its $m$ edges is incident to another node of this
  iteration is less than $0.01$.

  Consider any of the $m$ edges incident to a given node of this
  iteration. For each edge, we say it is \emph{good} if it survives in
  both copies of the graph, and is incident to an identified node from
  a previous iteration. Thus, the probability that an edge is good is
  at least $s^2 \cdot (0.99 \times 0.92)$.  Since $ms^2 > 22$, the
  expected number of good edges is greater than 20. The node will be
  identified if at least $8$ of its edges are good; applying the
  Chernoff bound, the probability that a given node is unidentified is
  at most $\exp(-3.606) < 0.02717$.

  Since this is true for each node of this iteration, regardless of
  the outcomes for previous nodes of the iteration, we can apply
  concentration inequalities even though the events are not
  independent. In particular, the number of identified nodes
  stochastically dominates the number of successes in $n^{0.25}$
  independent Bernoulli trials with probability $1 - \exp(-3.606)$
  (see, for example, Theorem 1.2.17 of \cite{MS}). Again
  applying the Chernoff Bound, the probability that the fraction of
  unidentified nodes exceeds $0.03$ is at most $\exp(0.27 n^{0.25} *
  0.01/4)$, which is negligible. To complete the induction, we need to
  show that the fraction of total degree incident to unidentified
  nodes is at most $0.08$. To observe this, note that the increase in
  degree is $2m n^{0.25}$; the unidentified fraction increases if the
  new nodes are unidentified (but we have seen the expected
  contribution here is at most $0.02717 m n^{0.25}$), or if the
  ``other'' endpoint is another node of this iteration (at most $0.01
  m n^{0.25}$), or if the ``other'' endpoint is an unidentified node
  (in expectation, at most $0.08 m n^{0.25}$). Again, a simple
  concentration argument completes the proof.
\end{proof}

\section{Experiments}\label{sec:experiments}

In this section we analyze the performance of our algorithm in
different experimental settings.  The main goal of this section is to
answer the following eight questions:
\begin{itemize}
\itemsep0pt
\item Are our theorems robust? Do our results depend on the constants
  that we use or are they more general?
\item How does the algorithm scale on very large graphs?
\item Does our algorithm work only for an underlying ``real'' network
  generated by a random process such as Preferential Attachment, or
  does it work for real social networks?
\item How does the algorithm perform when the two networks to be
  matched are not generated by independently deleting edges, but by a
  different process like a cascade model?
\item How does the algorithm perform when the two networks to be
  matched have different scopes? Is the algorithm robust to highly
  correlated edge deletion?
\item Does our model capture reality well? In more realistic
  scenarios, with distinct but similar graphs, does the algorithm
  perform acceptably?
\item How does our algorithm perform when the network is under attack?
  Can it still have high precision? Is it easy for an adversary to trick our
  algorithm?
\item How important is it to bucket nodes by degree? How big is the
  impact on the algorithm's precision? How does our algorithm compare
  with a simple algorithm that just counts the number of common
  neighbors? 
\end{itemize}

To answer these eight questions, we designed $4$ different experiments
using $6$ different publicly available data sets. These experiments
are increasingly challenging for our algorithm, which performs well in
all cases, showing its robustness.  Before entering into the details
of the experiments, we describe briefly the basic datasets used in the
paper. We use synthetic random graphs generated by the Preferential
Attachment~\cite{ba99}, Affiliation Network~\cite{ls09}, and
RMAT~\cite{CZF04} processes; we also consider an early snapshot of the
Facebook graph~\cite{vmcg09}, a snapshot of DBLP~\cite{dblp}, the
email network of Enron~\cite{KY04}, a snapshot of Gowalla~\cite{CML11}
(a social network with location information), and Wikipedia in two
languages~\cite{wiki}. In Table~\ref{tab:description} we report some
general statistics on the networks.

\begin{table}[t!]
\begin{center}
\begin{tabular}{|*{3}{c|}}
\hline
Network & Number of nodes   & Number of edges \\ \hline 
PA~\cite{ba99} & 1,000,000 & 20,000,000  \\ \hline
RMAT24~\cite{CZF04} & 8,871,645 & 520,757,402  \\ \hline
RMAT26~\cite{CZF04} & 32,803,311 & 2,103,850,648  \\ \hline
RMAT28~\cite{CZF04} & 121,228,778 & 8,472,338,793  \\ \hline
AN~\cite{ls09} & 60,026 & 8,069,546  \\ \hline
Facebook~\cite{vmcg09} & 63,731 & 1,545,686  \\ \hline
DBLP~\cite{dblp} & 4,388,906 & 2,778,941  \\ \hline
Enron~\cite{KY04} & 36,692 & 367,662  \\ \hline
Gowalla~\cite{CML11} & 196,591 & 950,327  \\ \hline
French Wikipedia~\cite{wiki} & 4,362,736 & 141,311,515  \\ \hline
German Wikipedia~\cite{wiki}  & 2,851,252 &  81,467,497 \\ \hline
\end{tabular}
\caption{The original 11 datasets.}
\label{tab:description}
\end{center}
%\vspace{-0.25in}
\end{table}

\noindent \textbf{Robustness of our Theorems:} To answer the first question, we
use as an underlying graph the preferential attachment graph described
above, with 1,000,000 nodes and $m = 20$. We analyze the performance of
our algorithm when we delete edges with probability $s=0.5$ and with
different seed link probabilities.  The main goal of this experiment
is to show that the values of $m,s$ needed in our proof are only
required for the calculations; the algorithm is effective even with
much lower values. With the specified parameters, for the majority of
nodes, the expected number of neighbors in the intersection of both
graphs is $5$. Nevertheless, as shown in Figure~\ref{fig:pa}, our
algorithm performs remarkably well, making zero errors regardless of
the seed link probability. Further, it recovers almost the entire
graph. Unsurprisingly, lowering the threshold for our algorithm
increases recall, but it is interesting to note that in this setting,
it does not affect precision at all.

\begin{figure}
\begin{center}
\includegraphics[width=75mm]{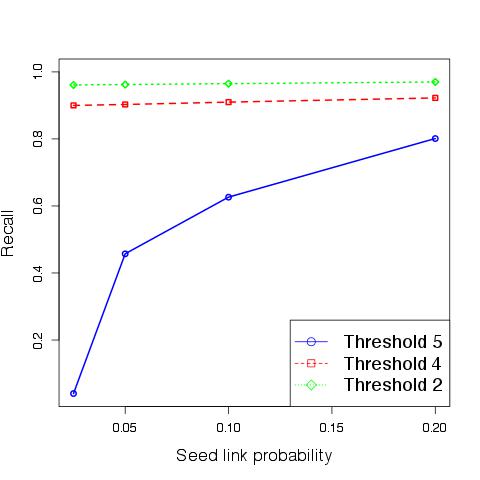}
\end{center}
%\vspace{-0.2in}
\caption{The number of corrected pairs detected with different threshold for
the preferential attachment model with random deletion. The precision is not
shown in the plot because it is always 100\%.}
\label{fig:pa}
\end{figure}

\noindent \textbf{Efficiency of our algorithms:} Here we tested our
algorithms with datasets of increasing size. In particular we generate
3 synthetic random graphs of increasing size using the RMAT random
model. Then we use the three graphs as the underlying ``real''
networks and we generate $6$ graphs from them with edges surviving
with probability $0.5$. Finally we analyze the running time of our
algorithm with seed link probability equal to $0.10$. As shown in
Table~\ref{tab:running_time}, using the same amount of resources, the
running time of the algorithm increases by at most a factor $12.544$
between the smallest and the largest graph.

\begin{table}[t!]
\begin{center}
\small
\begin{tabular}{|*{3}{c|}}
\hline
Network & Number of nodes   & Relative running time \\ \hline 
RMAT24 & 8871645 & 1  \\ \hline
RMAT26 & 32803311 & 1.199  \\ \hline
RMAT28 & 121228778 & 12.544  \\ \hline
\end{tabular}
\caption{The relative running time of the algorithm on three RMAT graphs as a function of numbers
of nodes in the graph.}
\label{tab:running_time}
\end{center}
%\vspace{-0.2in}
\end{table}
\noindent \textbf{Robustness to other models of the underlying graph:}
For our third question, we move away from synthetic graphs, and
consider the snapshots of Facebook and the Enron email networks as our
initial underlying networks. For Facebook, edges survive either with
probability $s=0.5$ or $s=0.75$, and we analyze performance of our
algorithm with different seed link probabilities. For Enron, which is
a much sparser network, we delete the edges with probability $s=0.5$
and analyze performance of our algorithm with seed link probability
equal to $0.10$.  The main goal of these experiments is to show that
our algorithm has good performance even outside the boundary of our
theoretical results even when the underlying network is not generated
by a random model.

In the first experiment with Facebook, when edges survive with
probability $0.75$, there are $63584$ nodes with degree at least $1$
in both networks.\footnote{Note that we can only detect nodes which
  have at least degree 1 in both networks} In the second, with edges
surviving with probability $0.5$, there are $62854$ nodes with this
property. In this case, the results are also very strong; see
Table~\ref{table:fb}. Roughly $28\%$ of nodes have extremely low
degree $(\le 5)$, and so our algorithm cannot obtain recall as high as
in the previous setting. However, we identify a very large fraction of
the roughly $45250$ nodes with degree above $5$, and the precision is
still remarkably good; in all cases, the error is well under
$1\%$. Table 2 presents the full results for the harder case, with
edge survival probability $0.5$. With edge survival probability $0.75$
(not shown in the table), performance is even better: At threshold $2$
and the lowest seed link probability of $5\%$, we correctly identify
$46626$ nodes and incorrectly identify $20$, an error rate of well
under $0.05\%$.
In the case of Enron, the original email network is very sparse, with
an average degree of approximately 20; this means that each copy has
average degree roughly 10, which is much sparser than real social
networks. Of the 36,692 original nodes, only 21,624 exist in the
intersection of the two copies; over 18,000 of these have degree $\le
5$, and the average degree is just over $4$. Still, with matching
threshold $5$, we identify almost all the nodes of degree $5$ and
above, and even in this very sparse graph, the error rate among newly
identified nodes is $4.8\%$.

\begin{table}[t!]
\begin{center}
\small
\begin{tabular}{|*{7}{c|}}
\hline
Pr & \multicolumn{2}{|c|}{Threshold 5}   & \multicolumn{2}{|c|}{Threshold 4}& \multicolumn{2}{|c|}{Threshold 2} \\ \hline 
& Good & Bad & Good & Bad & Good & Bad\\ \hline
$20\%$ & 23915 & 0 & 28527 & 53 & 41472 & 203 \\ \hline
$10\%$ & 23832 & 49 & 32105 & 112 & 38752 & 213 \\ \hline
$5\%$ & 11091 & 43 & 28602 & 118 & 36484 & 236 \\ \hline\hline
\end{tabular}
\begin{tabular}{|*{7}{c|}}
\hline
Pr & \multicolumn{2}{|c|}{Threshold 5}   & \multicolumn{2}{|c|}{Threshold 4}& \multicolumn{2}{|c|}{Threshold 3} \\ \hline 
& Good & Bad & Good & Bad & Good & Bad\\ \hline
$10\%$ & 3426 & 61 & 3549 & 90 & 3666 & 149 \\ \hline
\end{tabular}
\caption{Results for Facebook (Left) and Enron (Right) under the random deletion model. Pr
  denotes the seed link probability.}
\label{table:fb}
\end{center}
%\vspace{-0.25in}
\end{table}

\noindent \textbf{Robustness to different deletion models:} We now turn our
attention to the fourth question: How much do our results depend on the
process by which the two copies are generated? To answer this, we
analyze a different model where we generate the two copies of the
underlying graph using the \emph{ Independent Cascade Model}
of~\cite{GLM01}. More specifically, we construct a graph starting from
one seed node in the underlying social network and we add to the graph
the neighbors of the node with probability $p=0.05$.  Subsequently,
every time we add a node, we consider all its neighbors and add each
of them independently with probability $p=0.05$ (note that we can try
to add a node to the graph multiple times).

The results in this cascade model are extremely good; in fact, for
both Facebook and Enron we have $0$ errors; as shown for Facebook in
Figure~\ref{fig:fbc}, we are able to identify almost all the nodes in
the intersection of the two graphs (even at seed link prob. $5\%$, we
identify $16,273/16533 = 98.4\%$). 

\begin{figure}
\begin{center}
\includegraphics[width=75mm]{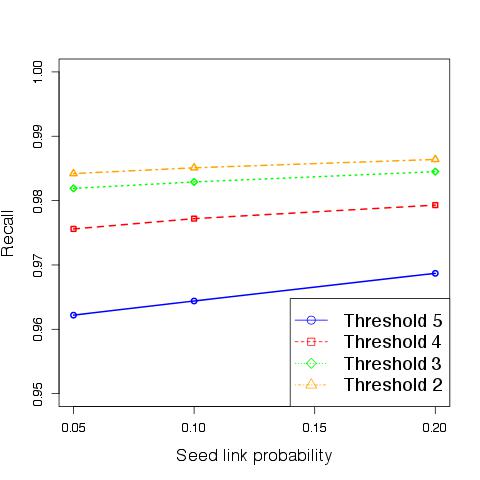}
\end{center}
%\vspace{-0.25in}
\caption{The number of corrected pairs detected with different threshold for
the two Facebook graphs generated by the Independent Cascade
Model. The plot does not show precision, since it is always 100\%.}
\label{fig:fbc}
\end{figure}
 
\noindent \textbf{Robustness to correlated edge deletion:} We now
analyze one of the most challenging scenarios for our algorithm where,
independently in the two realizations of the social network, we delete
all or none of the edges in a community. For this purpose, we consider
the Affiliation Networks model~\cite{ls09} as the underlying real
network. In this model, a bipartite graph of users and interests is
constructed using a preferential attachment-like process and then two
users are connected in the network if and only if they share an
interest (for the model details, refer to~\cite{ls09}). To generate the
two copies in our experiment, we delete the interests independently in
each copy with probability $0.25$, and then we generate the graph
using only the surviving interests.  Note that in this setting, the
same node in the two realizations can have very different neighbors.
Still, our algorithm has very high precision and recall, as shown in
Table~\ref{tb:an}.

\begin{table}[t!]
\begin{center}
\small
\begin{tabular}{|*{7}{c|}}
\hline
Pr &  \multicolumn{2}{|c|}{Threshold 4}&
\multicolumn{2}{|c|}{Threshold 3} & \multicolumn{2}{|c|}{Threshold 2} \\ \hline 
& Good & Bad & Good & Bad  & Good & Bad\\ \hline
$10\%$ & 54770 & 0 & 55863 & 0 & 55942 & 0\\ \hline
\end{tabular}
\caption{Results for the Affiliation Networks model under correlated edge deletion probability.}
\label{tb:an}
\end{center}
%\vspace{-0.25in}
\end{table}
 
\noindent \textbf{Real world scenarios:} Now we move to the most
challenging case, where the two graphs are no longer generated by a
mathematical process that makes 2 imperfect copies of the same
underlying network. For this purpose, we conduct two types of
experiments. First, we use the DBLP and the Gowalla datasets in which
each edge is annotated with a time, and construct 2 networks by taking
edges in disjoint time intervals.  Then we consider the French- and 
German-language Wikipedia link graph.

From the co-authorship graph of DBLP, the first network is generated by considering only the
publications written in even years, and the second is generated by
considering only the publications written in odd years. Gowalla is a
social network where each user could also check-in to a location (each
check-in has an associated timestamp).  Using this information we
generate two Gowalla graphs; in the first graph, we have an edge
between nodes if they are friends and if and only if they check-in to
approximately the same location in an odd month. In the second, we
have an edge between nodes if they are friends and if and only if they
check-in in approximately the same location in an even month.

Note that for both DBLP and Gowalla, the two constructed graphs have a
different set of nodes and edges, with correlations different from the
previous independent deletion models. Nevertheless we will see that
the intersection is big enough to retrieve a good part of the
networks.

In DBLP, there are $380,129$ nodes in the intersection of the two
graphs, but the considerable majority of them have extremely low
degree. Over $310$K have degree less than $5$ in the
intersection of the two graphs, and so again we cannot hope for
extremely high recall. However, we do find considerably more nodes
than in the input set. We start with a $10\%$ probability of seed
links, resulting in $32087$ seeds; however, note that most of these
have extremely low degree, and hence are not very useful. As shown in
table~\ref{table:dblp}, we have nearly $69,000$ nodes identified, with
an error rate of under $4.17\%$. Note that we identify over half the
nodes of degree at least $11$, and a considerably larger fraction of
those with higher degree. We include a plot showing precision and
recall for nodes of various degrees (Figure~\ref{fig:prdblpgowala}).

For Gowalla, there are 38103 nodes in the intersection of the two
graphs, of which over 32K have degree $\le 5$. We start with 3800
seeds, of which most are low-degree and hence not useful. We identify
over 4000 of the (nearly 6000) nodes of degree above 5, with an error
rate of $3.75\%$. See Table~\ref{table:dblp} and
Figure~\ref{fig:prdblpgowala} for more details.

Finally for a still more challenging scenario, we consider a case
where the 2 networks do not have any common source, but yet may have
some similarity in their structure. In particular, we consider the
case of the French- and German-language Wikipedia sites, which have
4.36M and 2.85M nodes respectively. Wikipedia also maintains a set of
inter-language links, which connect corresponding articles in a pair
of languages; for French and German, there are 531710 links,
corresponding to only $12.19\%$ of the French articles. The relatively
small number of links illustrates the extent of the difference between
the French and German networks. Starting with $10\%$ of the
inter-language links as seeds, we are able to nearly triple the number
of links (including finding a number of new links not in the input
inter-language set), with an error rate of $17.5\%$ in new
links. However, some of these mistakes are due to human errors in
Wikipedia's inter-language links, while others mistake French articles
to closely connected German ones; for instance, we link the French
article for Lee Harvey Oswald (the assassin of President Kennedy) to
the German article on the assassination.

\begin{table}[t!]
\begin{center}
\small
\begin{tabular}{|l|l|l|l|l|l|l|}
\hline
Pr & \multicolumn{2}{|c|}{Threshold 5}   & \multicolumn{2}{|c|}{Threshold 4}& \multicolumn{2}{|c|}{Threshold 2} \\ \hline 
& Good & Bad & Good & Bad & Good & Bad\\ \hline
10 & 42797 & 58 & 53026 & 641 & 68641 & 2985 \\ \hline\hline
\end{tabular}
\begin{tabular}{|l|l|l|l|l|l|l|}
\hline
Pr & \multicolumn{2}{|c|}{Threshold 5}   & \multicolumn{2}{|c|}{Threshold 4}& \multicolumn{2}{|c|}{Threshold 2} \\ \hline 
& Good & Bad & Good & Bad & Good & Bad\\ \hline
10 & 5520 & 29 & 5917 & 48 & 7931 & 155 \\ \hline\hline
\end{tabular}
\begin{tabular}{|l|l|l|l|l|l|l|}
\hline
Pr & \multicolumn{2}{|c|}{Threshold 5}   & \multicolumn{2}{|c|}{Threshold 3} \\ \hline 
& Good & Bad &  Good & Bad \\ \hline
10 & 108343 & 9441 & 122740 & 14373 \\ \hline
\end{tabular}
\caption{Results for DBLP (Top Left), Gowala (Top Right), and Wikipedia (Bottom)}
\label{table:dblp}
\end{center}
%\vspace{-0.25in}
\end{table}

\begin{figure}[h]
\centering
%\vspace{-8pt}
\includegraphics[width=0.45\textwidth]{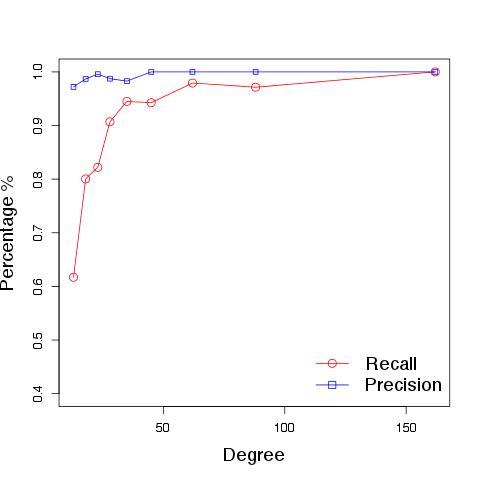}       \hspace{0.01\textwidth}            
\includegraphics[width=0.45\textwidth]{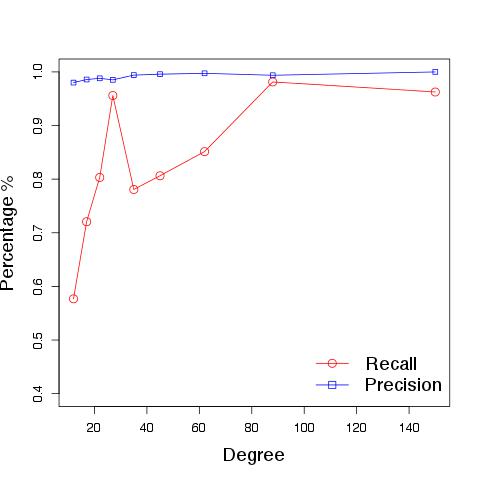}
\caption{\small Precision and Recall vs. Degree Distribution for
  Gowala (left) and DBLP (right).}
 \label{fig:prdblpgowala}
%\vspace{-3pt}
\end{figure}

\noindent \textbf{Robustness to attack:} We now turn our attention to
a very challenging question: what is the performance of our
algorithm when the network is under attack? In order to answer this
question, we again consider the Facebook network as the underlying
social network, and from it we generate two realizations with edge
probability $0.75$. Then, in order to simulate an attack, in each
network for each node $v$ we create a malicious copy of it, $w$, and
for each node $u$ connected to $v$ in the network (that is, $u\in
N(v)$), we add the edge $(u,w)$ independently with probability
$0.5$. Note that this is a very strong attack model (it assumes that
users will accept a friend request from a 'fake' friend with
probability 0.5), and is designed to circumvent our matching
algorithm. Nevertheless when we run our algorithm with seed link
probability equal to $0.1$, and with threshold equal to $2$ we notice
that we are still able to align a very large fraction of the two
networks with just a few errors ($46955$ correct matches and $114$ wrong
matches, out of $63731$ possible good matches).

\noindent \textbf{Importance of degree bucketing, comparison with
  straightforward algorithm:} We now consider our last question: How
important is it to bucket nodes by degree? How big is the impact on
the algorithm's precision? How does our algorithm compare with a
straightforward algorithm that just counts the number of common
neighbors? To answer this question, we run a few experiments. First,
we consider the Facebook graph with edge survival probability $0.5$
and seed link probability $5\%$, and we repeat the experiments again
without using the degree bucketing and with threshold equal 1.  In
this case we observe that the number of bad matching increases by a
factor of $50\%$ without any significant change in the number of good
matchings.

Then we consider other two interesting scenarios: How does this simple
algorithm perform on Facebook under attack? And how does it perform on
matching Wikipedia pages?
Those two experiments show two weaknesses of this simple
algorithm. More precisely, in the first case the simple algorithm
obtains $100\%$ precision but its recall is very low.  It is indeed
able to reconstruct less than half of the number of matches found by
our algorithm (22346 vs 46955). On the other hand, the second setting
shows that the precision of this simple algorithm can be very
low. Specifically, the error rate of the algorithm is $27.87\%$, while
our algorithm has error rate only $17.31\%$. In this second setting
(for Wikipedia) the recall is also very low, less than $13.52\%$;
there are 71854 correct matches, of which most (53174) are seed links,
and 7216 wrong matches.

\section{Conclusions}

In this paper, we present the first provably good algorithm for social
network reconciliation. We show that in well-studied models of social
networks, we can identify almost the entire network, with no
errors. Surprisingly, the perfect precision of our algorithm holds
even experimentally in synthetic networks. For the more realistic data
sets, we still identify a very large fraction of the nodes with very
low error rates. Interesting directions for future work include
extending our theoretical results to more network models and
validating the algorithm on different and more realistic data sets.\\

\noindent\textbf{Acknowledgements}

\noindent We would like to thank Jon Kleinberg for useful discussions and
Zolt\'an Gy\"ongyi for suggesting the problem.

\bibliographystyle{abbrv}
\bibliography{contact-matching}

\end{document}